\documentclass{llncs}
\usepackage{amsfonts}
\usepackage{amsmath}
\usepackage{amssymb}
\usepackage{algorithmicx}
\usepackage{algpseudocode}
\usepackage{graphicx}
\usepackage{gastex}
\usepackage{array}

\usepackage{listings}
\usepackage{hyperref}

\newcommand{\lvec}[1]{\overset{{}_{\leftarrow}}{#1}}
\def\add{{\sf add}}
\def\eertree{{\sf eertree}}
\def\maxsuf{{\sf maxSuf}}
\def\maxpref{{\sf maxPref}}
\def\len{{\sf len}}
\def\link{{\sf link}}
\def\go{{\sf to}}
\def\occ{{\sf occ}}
\def\occas{{\sf occAsMax}}
\def\argmax{{\sf argmax}}
\def\sufc{{\sf sufCount}}
\def\prefc{{\sf prefCount}}
\def\flag{{\sf flag}}
\def\pop{{\sf pop}}
\def\quick{{\sf quickLink}}
\def\series{{\sf seriesLink}}
\def\direct{{\sf directLink}}

\def\addv{{\sf addVersion}}
\def\search{{\sf searchTree}}
\def\pred{{\sf pred}}
\def\symb{{\sf symb}}
\def\ans{{\sf ans}}
\def\diff{{\sf diff}}
\def\getmin{{\sf getMin}}
\def\res{{\sf res}}
\def\ddp{{\sf dp}}
\def\newv{{\sf newNode}}
\def\nod{{\sf node}}
\def\cT{{\cal T}}

\begin{document}

\thispagestyle{empty}

\title{EERTREE: An Efficient Data Structure for Processing Palindromes in Strings} 
\author{{Mikhail Rubinchik and Arseny M. Shur \\Ural Federal University}}%

\author{Mikhail Rubinchik \and Arseny M. Shur}

\institute{Ural Federal University, Ekaterinburg, Russia\\ \email{mikhail.rubinchik@gmail.com}, \email{arseny.shur@urfu.ru}}
\maketitle

\vspace*{-4mm}
\begin{abstract}
We propose a new linear-size data structure which provides a fast access to all palindromic substrings of a string or a set of strings. This structure inherits some ideas from the construction of both the suffix trie and suffix tree. Using this structure, we present simple and efficient solutions for a number of problems involving palindromes.\\[2pt]
\end{abstract}

\section{Introduction}

Palindromes are one of the most important repetitive structures in strings. During the last decades they were actively studied in formal language theory, combinatorics on words and stringology. Recall that a palindrome is any string $S = a_1a_2\cdots a_n$ equal to its reversal $\lvec{S} = a_n\cdots a_2a_1$.

There are a lot of papers concerning the palindromic structure of strings. The most important problems in this direction include the search and counting of palindromes in a string and the factorization of a string into palindromes. Manacher \cite{Man75} came up with a linear-time algorithm which can be used to find all maximal palindromic substrings of a string, along with its palindromic prefixes and suffixes. The problem of counting and listing distinct palindromic substrings was solved offline in \cite{GPR10} and online in \cite{KRS13}. Knuth, Morris, and Pratt \cite{KMP77} gave a linear-time algorithm for checking whether a string is a product of even-length palindromes. Galil and Seiferas \cite{GaSe78} asked for such an algorithm for the \emph{$k$-factorization} problem: decide whether a given string can be factored into exactly $k$ palindromes, where $k$ is an arbitrary constant. They presented an online algorithm for $k=1,2$ and an offline one for $k=3,4$. An online algorithm working in $O(kn)$ time for the length $n$ string and any $k$ was designed in \cite{KRS15}. Close to the $k$-factorization problem is the problem of finding the \emph{palindromic length} of a string, which is the minimal $k$ in its $k$-factorization. This problem was solved by Fici et al. in $O(n\log n)$ time \cite{FGKK14}.
In this paper we present a new tree-like data structure, called eertree\footnote{This structure can be found, with the reference to the first author, in a few IT blogs under the name ``palindromic tree''. See, e.g., \texttt{http://adilet.org/blog/25-09-14/}.}, which allows one to simplify and speed up solutions to search, counting and factorization problems as well as to several other palindrome-related algorithmic problems. This structure can also cope with Watson--Crick palindromes \cite{KaMa10} and other palindromes with involution and may be interesting for the RNA studies along with the affix trees \cite{MaPa05} and affix arrays \cite{Str07}. 

In Sect.~\ref{basicPart} we first recall the problem of counting distinct palindromic substrings in an online fashion. This was a motive example for inventing eertree. This data structure contains the digraph of all palindromic factors of an input string $S$ and supports the operation $\add(c)$ which appends a new symbol to the end of $S$. Thus, the number of nodes in the digraph equals the number of distinct palindromes inside $S$. Maintaining an eertree for a length $n$ string with $\sigma$ distinct symbols requires $O(n\log\sigma)$ time and $O(n)$ space (for a random string, the expected space is $O(\sqrt{n\sigma})$). After introducing the eertree we discuss some of its properties and simple applications.

In Section~\ref{applications} we study  advanced questions related to eertrees. We consider joint eertree of several strings and name a few problems solved with its use. Then we design two ``smooth'' variations of the algorithm which builds eertree. These variations require at most logarithmic time for each call of $\add(c)$ and then allow one to support an eertree for a string with two operations: appending and deleting the last symbol. Using one of these variations, we design a fast backtracking algorithm enumerating all \emph{rich} strings over a fixed alphabet up to a given length. (A string is rich if it contains the maximum possible number of distinct palindromes.) Finally, we show that eertree can be efficiently turned into a persistent data structure.
 
The use of eertrees for factorization problems is described in Sect.~\ref{splittingSection}. Namely, new fast algorithms are given for the $k$-factorization of a string and for computing its palindromic length. We also conjecture that the palindromic length can be found in linear time and provide some argument supporting this conjecture.

\paragraph*{Definitions and Notation.}

We study finite strings, viewing them as arrays of symbols: $w=w[1..n]$. The notation $\sigma$ stands for the number of distinct symbols of the processed string. We write $\varepsilon$ for the empty string, $|w|$ for the length of $w$, $w[i]$ for the $i$th letter of $w$ and $w[i..j]$ for $w[i]w[i{+}1]\ldots w[j]$, where $w[i..i{-}1] = \varepsilon$ for any~$i$. A string $u$ is a \emph{substring} of $w$ if $u=w[i..j]$ for some $i$ and $j$. A substring $w[1..j]$ (resp., $w[i..n]$) is a \emph{prefix} [resp. \emph{suffix}] of $w$.  If a substring (prefix, suffix) of $w$ is a palindrome, it is called a \emph{subpalindrome} (resp. \emph{prefix-palindrome},\emph{ suffix-palindrome}). A subpalindrome $w[l..r]$ has \emph{center} $(l{+}r)/2$, and \emph{radius} $\lceil (r{-}l{+}1)/2 \rceil$. Throughout the paper, we do not count $\varepsilon$ as a palindrome. 

\emph{Trie} is a rooted tree with some nodes marked as terminal and all edges labeled by symbols such that no node has two outgoing edges with the same label. Each trie represents a finite set of strings, which label the paths from the root to the terminal nodes.

\section{Building An Eertree}\label{basicPart}

\subsection{Motive problem: distinct subpalindromes online}\label{motivation}

Well known online linear-time Manacher's algorithm \cite{Man75} outputs maximal radiuses of subpalindromes in a string for all possible centers, thus encoding all subpalindromes of a string. Another interesting problem is to find and count all distinct subpalindromes. Groult et al. \cite{GPR10} solved this problem offline in linear time and asked for an online solution. Such a solution in $O(n \log \sigma)$ time and $O(n)$ space was given in \cite{KRS13}, based on Manacher's algorithm and Ukkonen's suffix tree algorithm \cite{Ukk95}. As was proved in the same paper, this solution is asymptotically optimal in the comparison-based model. But in spite of a good asymptotics, this algorithm is based on two rather ``heavy'' data structures. In is natural to try finding a lightweight structure for solving the analyzed problem with the same asymptotics. Such a data structure, eertree, is described below. Its further analysis revealed that it is suitable for coping with many algorithmic problems involving palindromes.

\subsection{Eertree: structure, interface, construction}

The basic version of eertree supports a single operation $\add(c)$, which appends the symbol $c$ to the processed string (from the right), updates the data structure respectively, and returns the number of new palindromes appeared in the string. According to the next lemma, $\add(c)$ returns 0 or 1.

\begin{lemma}[{\!\cite{DJP01}}] \label{nnodes}
Let $S$ be a string and $c$ be a symbol. The string $Sc$ contains at most one subpalindrome which is not a substring of $S$. This new palindrome is the longest suffix-palindrome of $Sc$.
\end{lemma}


From inside, eertree is a directed graph with some extra information. Its nodes, numbered with positive integers starting with 1, are in one-to-one correspondence with subpalindromes of the processed string. Below we denote a node and the corresponding palindrome by the same letter. We write $\eertree(S)$ for the state of eertree after processing the string $S$ letter by letter, left to right.

\begin{remark}
To report the number of distinct subpalindromes of $S$, just return the maximum number of a node in $\eertree(S)$.
\end{remark}

Each node $v$ stores the length $\len[v]$ of its palindrome. For the initialization purpose, two special nodes are added: with the number $0$ and length $0$ for the empty string, and with the number $-1$ and length $-1$ for the ``imaginary string''. 

The edges of the graph are defined as follows. If $c$ is a symbol, $v$ and $cvc$ are two nodes, then an edge labeled by $c$ goes from $v$ to $cvc$. The edge labeled by $c$ goes from the node $0$ (resp. $-1$) to the node labeled by $cc$ (resp., by $c$) if it exists. This explains why we need two initial nodes. The outgoing edges of a node $v$ are stored in a dictionary which, given a symbol $c$, returns the edge $\go[v][c]$ labeled by it. Such a dictionary is implemented as a binary balanced search tree. 

An unlabeled \emph{suffix link} $\link[u]$ goes from $u$ to $v$ if $v$ is the longest proper suffix-palindrome of $u$. By definition, $\link[c]=0$, $\link[0]=\link[{-}1]={-}1$. The resulting graph, consisting of nodes, edges, and suffix links, is the eertree; see Fig.~\ref{fig:eertree} for an example. 

\begin{figure}[hbt!]
\vspace*{-4mm}
\centering
\includegraphics[trim=147 546 330 122,clip]{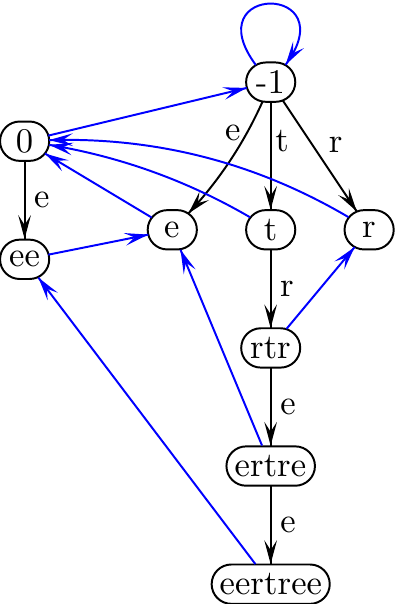}
\vspace*{-1mm}
\caption{The eertree of the string $eertree$. Edges are black, suffix links are blue.}
\label{fig:eertree}
\end{figure}


\begin{lemma}\label{oneEdge}
A node of positive length in an eertree has exactly one incoming edge.
\end{lemma}

\begin{proof}
An edge leading to a node $u$ is labeled by $c=u[1]$. Then its origin must be the node $v$ such that $u=cvc$ or the node $-1$ if $u=c$.\qed
\end{proof}

\begin{proposition} \label{spacen}
The eertree of a string $S$ of length $n$ is of size $O(n)$.
\end{proposition}

\begin{proof}
The eertree of $S$ has at most $n{+}2$ nodes, including the special ones (by Lemma~\ref{nnodes}), at most $n$ edges (by Lemma~\ref{oneEdge}), and at most $n{+}2$ suffix links (one per node).\qed
\end{proof}

\begin{proposition} \label{nlogsigma}
For a string $S$ of length $n$, $\eertree(S)$ can be built online in $O(n\log\sigma)$ time.
\end{proposition}

\begin{proof}
We start defining $\eertree(\varepsilon)$ as the graph with two nodes ($0$ and $-1$) and two suffix links. Then we make the calls $\add(S[1]),\ldots,\add(S[n])$ in this order. By Lemma~\ref{nnodes} and the definition of $\add$, after each call we know the longest suffix-palindrome $\maxsuf(T)$ of the string $T$ processed  so far. We support the following invariant: after a call to $\add$, all edges and suffix links between the existing nodes are defined. In this case, adding a new node $u$ one must build exactly one edge (by Lemma~\ref{oneEdge}) and one suffix link: any suffix-palindrome of $u$ is its prefix as well, and hence the destination node of the suffix link from $u$ already exists.

Consider the situation after $i$ calls. We have to perform the next call, say $\add(a)$, to $T=S[1..i]$. We need to find the maximum suffix-palindrome $P$ of $Ta$. Clearly, $P=a$ or $P=aQa$, where $Q$ is a suffix-palindrome of $T$. Thus, to determine $P$ we should find the longest suffix-palindrome of $T$ preceded by $a$. To do this, we traverse the suffix-palindromes of $T$ in the order of decreasing length, starting with $\maxsuf(T)$ and following suffix links. For each palindrome we read its length $k$ and compare $T[i{-}k]$ against $a$ until we get an equality or arrive at the node $-1$. In the former case, the current palindrome is $Q$; we check whether it has an outgoing edge labeled by $a$. If yes, the edge leads to $aQa=P$, and $P$ is not new; if no, we create the node $P$ of length $|Q|+2$ and the edge $(Q,P)$. In the latter case, $P=a$; as above, we check the existence of $P$ in the graph from the current node (which is now $-1$) and create the node if necessary, together with the edge $(-1,P)$ and the suffix link $(P,0)$. 

It remains to create the suffix link from $P$ if $|P|>1$. It leads to the second longest suffix-palindrome of $Ta$. This palindrome can be found similar to $P$: just continue traversing suffix-palindromes of $T$ starting with the suffix link of $Q$. 

Now estimate the time complexity. During a call to $\add(a)$, one checks the existence of the edge from $Q$ with the label $a$ in the dictionary, spending $O(\log\sigma)$ time. The path from the old to the new value of $\maxsuf$ requires one transition by an edge (from $Q$ to $P$) and $k\ge0$ of transitions by suffix links, and is accompanied by $k{+}1$ comparisons of symbols. In order to estimate $k$, follow the position of the first symbol of $\maxsuf$: a transition by a suffix link moves it to the right, and a transition by an edge moves it one symbol to the left. During the whole process of construction of $\eertree(S)$, this symbol moves to the right by $\le n$ symbols. Hence, the total number of transitions by suffix links is $\le 2n$. The same argument works for the second longest suffix-palindrome, which was used to create suffix links. Thus, the total number of graph transitions and symbol comparisons is $O(n)$, and the time complexity is dominated by checking the existence of edges, $O(n\log\sigma)$ time in total.\qed
\end{proof}

\subsection{Some properties of eertrees} \label{ssec:prop}

We call a node \emph{odd} (resp., even) if it corresponds to an odd-length (resp., even-length) palindrome. By \emph{suffix path} we mean a path consisting of suffix links. 

\begin{lemma}\label{twotree}
\emph{1)} Nodes and edges of an eertree form two weakly connected components: the tree of odd (resp., of even) nodes rooted at $-1$ (resp., at $0$).\\
\emph{2)} The tree of even (resp., odd) nodes is precisely the trie of right halves of even-length palindromes (resp., the trie of right halves, including the central symbol, of odd-length palindromes).\\
\emph{3)} Nodes and inverted suffix links of an eertree form a tree with a loop at its root $-1$.
\end{lemma}

\begin{proof}
1) If an edge $(u,v)$ exists, then $|v|=|u|+2$. Hence, the edges of an eertree constitute no cycles, odd nodes are unreachable from even ones, and vice versa. Further, Lemma~\ref{oneEdge} implies that each even (resp., odd) node can be reached by a unique path from $0$ (resp., $-1$). So we have the two required trees.

2) This is immediate from the definitions of a trie and an edge.

3) The suffix link decreases the length of a node, except for the node $-1$. So the only cycle of suffix links is the loop at $-1$. Each node has a unique suffix link and is connected by a suffix path to the node $-1$. So the considered graph is a tree (with a loop on the root) by definition. 
\end{proof}


\begin{remark}
Tries are convenient data structures, but a trie built from the set of all suffixes (or all factors) of a length $n$ string is usually of size $\Omega(n^2)$. For a linear-space implementation, such a  trie should be compressed into a more complicated and less handy structure: suffix tree or suffix automaton (DAWG). On the other hand, eertrees are linear-size tries and do not need any compression. Moreover, the size of an eertree is usually much smaller than $n$, because the expected number of distinct palindromes in a length $n$ string is $O(\sqrt{n\sigma})$ \cite{RuSh15}. This fact explains high efficiency of eertrees in solving different problems.
\end{remark}

\begin{remark}
A \emph{$\theta$-palindrome} is a string $S=a_1\cdots a_n$ equal to $\theta(a_n\cdots a_1)$, where $\theta$ is a symbol-to-symbol function and $\theta^2$ is the identity (see, e.g., \cite{KaMa10}). Clearly, an eertree containing all $\theta$-palindromes of a string can be built in the way described in Proposition~\ref{nlogsigma} (the comparisons of symbols should take $\theta$ into account).
\end{remark}

\subsection{First applications}\label{palrefrain}

We demonstrate the performance of eertrees on two test problems taken from student programming contests. The first problem is {\sf Palindromic Refrain} \cite[Problem A]{Apio}, stated as follows: for a given string $S$ find a subpalindrome $P$ maximizing the value $|P|\cdot \occ(S,P)$, where $\occ(S,P)$ is the number of occurrences of $P$ in $S$. The solution to this problem, suggested by the jury of the contest, included a suffix data structure and Manacher's algorithm.

\begin{proposition}
{\sf Palindromic Refrain} can be solved by an eertree with the use of $O(n)$ additional time and space.
\end{proposition}

\begin{proof}
In order to find $\occ[v]$ for each node of $\eertree(S)$, we store an auxiliary parameter $\occas[v]$, which is the number of $i$'s such that $\maxsuf(S[1..i])=v$. This parameter is easy to compute online: after a call to $\add$, we increment $\occas$ for the current $\maxsuf$. After building $\eertree(S)$, we compute the values of $\occ$ as follows:
\begin{equation} \label{eq:occ}
\occ[v]=\occas[v]+\sum_{u: \link[u]=v} \occ[u]\,.
\end{equation}
Indeed, if $v$ is a suffix of $S[1..i]$ for some $i$, then either $v=\maxsuf(S[1..i])$ and this occurrence is counted in $\occas[v]$, or $v=\maxsuf(u)$ for some suffix-palindrome $u$ of $S[1..i]$; in the latter case, $\link[u]=v$, and this occurrence of $v$ is counted in $\occ[u]$.  To compute the values of $\occ$ in the order prescribed by \eqref{eq:occ}, one can traverse the tree of suffix links bottom-up:

\lstset{frame=tb,
  language=Pascal,
  aboveskip=3mm,
  belowskip=3mm,
  showstringspaces=false,
  columns=flexible,
  basicstyle=\ttfamily,
  numbers=none,
  keywordstyle=\bf,
  breaklines=true,
  breakatwhitespace=true
  tabsize=3
}

\begin{lstlisting}[escapeinside={/*@}{@*/}]
for (v = size; v /*@$\ge$@*/ 1; v--)
	occ[v] = occAsMax[v]
for (v = size; v /*@$\ge$@*/ 1; v--)
	occ[ link[v] ] += occ[v]
\end{lstlisting}

Here {\sf size} is the maximum number of a node in $\eertree(S)$. Note that the node $\link[v]$ always has the number less than $v$, because $\link[v]$ exists at the moment of creation of $v$.
After computing $\occ$ for all nodes, $P = \argmax(\occ[v] \cdot \len[v])$. \qed
\end{proof}

The second problem is {\sf Palindromic Pairs} \cite[Problem B]{PPair}: for a string $S$, find the number of triples $i, j, k$ such that $1 \le i \le j < k \le |S|$ and the strings $S[i..j]$, $S[j {+} 1..k]$ are palindromes.

\begin{proposition}
{\sf Palindromic Pairs} can be solved by an eertree with the use of $O(n\log\sigma)$ additional time and $O(n)$ space.
\end{proposition}

\begin{proof}
Let $\maxsuf[j]=\maxsuf(S[1..j])$ and $\sufc[v]$ be the number of suffix-palindromes of the subpalindrome $v$ of $S$, including $v$ itself. Note that $\sufc[v] = 1 + \sufc[ \link[v] ]$. Hence, $\sufc[v]$ can be stored in the node $v$ of $\eertree(S)$ and computed when this node is created. In addition, we memorize the values $\maxsuf[1],\ldots,\maxsuf[n]$ in a separate array. The number of palindromes ending in position $j$ of $S$ is the number of suffix-palindromes of $S[1..j]$ or of $\maxsuf(S[1..j])$. So this number equals $\sufc[\maxsuf[j]]$.

Further, let $\prefc[v]$ be the number of prefix-palindromes of $v$ and $\maxpref[j]$ be the longest prefix-palindrome of $S[j..n]$. The values of $\prefc$ and $\maxpref$ can be found when building $\eertree(\lvec{S})$\footnote{The strings $S$ and $\lvec{S}$ have exactly the same subpalindromes, so there is no need to build the second eertree. We just perform calls to $\add$ on $\eertree(S)$ and fill $\prefc$ and $\maxpref$.}. Similar to the above, the number of palindromes beginning in position $j$ of $S$ is $\prefc[\maxpref[j]]$. Note that all additional computations take $O(1)$ time for each call of $\add$, except for the second eertree, which requires $O(n\log\sigma)$ time.

For a fixed $j$, the number of triples $(i,j,k)$ defining a palindromic pair is the number of palindromes ending at position $i$ times the number of palindromes beginning at position $j{+}1$. Hence, the answer to the problem is
$$
\sum_{j=1}^{n-1} \sufc[\maxsuf[j]]\cdot \prefc[\maxpref[j{+}1]].
$$
Since this is also a linear-time computation, we are done with the proof.\qed
\end{proof}

\section{Advanced Modifications of Eertrees}\label{applications}

\subsection{Joint eertree for several strings}\label{common}

When a problem assumes the comparison of two or more strings, it may be useful to build a joint data structure. For example, a variety of problems can be solved by joint (``generalized'') suffix trees, see \cite{Gus97}.  Here we introduce the \emph{joint eertree} of a set of strings and name several problems it can solve.

A joint eertree $\eertree(S_1,\ldots,S_k)$ is built as follows. We build $\eertree(S_1)$ in a usual fashion; then reset the value of $\maxsuf$ to 0 and proceed with the string $S_2$, addressing the $\add$ calls to the currently built graph; and so on, until all strings are processed. Each created node stores an additional $k$-element boolean array $\flag$. After each call to $\add$, we update $\flag$ for the current $\maxsuf$ node, setting its $i$th bit to 1, where $S_i$ is the string being processed. As a result, $\flag[v][i]$ equals 1 if and only if $v$ is contained in $S_i$.

Some problems easily solved by a joint eertree are gathered below.\\[3pt]
\centerline{
\setlength{\extrarowheight}{3pt}
\begin{tabular}{c|c}
Problem&Solution\\
\hline
\begin{minipage}[t]{5.6cm}
\centering
Find the number of subpalindromes, common to all $k$ given strings.	
\end{minipage}
&
\begin{minipage}[t]{6.2cm}
Build $\eertree(S_1,\ldots,S_k)$ and count the nodes having only $1$'s in the $\flag$ array.
\end{minipage}\\[3pt]
\hline
\begin{minipage}[t]{5.6cm}
Find the longest subpalindrome contained in all $k$ given strings.	
\end{minipage}
&
\begin{minipage}[t]{6.2cm}
Build $\eertree(S_1,\ldots,S_k)$. Among the nodes having only $1$'s in the $\flag$ array, find the node of biggest length.
\end{minipage}\\[3pt]
\hline
\begin{minipage}[t]{5.6cm}
For strings $S$ and $T$ find the number of palindromes $P$ having more occurrences in $S$ than in $T$.
\end{minipage}
&
\begin{minipage}[t]{6.2cm}
Build $\eertree(S,T)$, computing $\occ_S$ and $\occ_T$ in its nodes (see {\sf Palindromic Refrain} in Sect.~\ref{palrefrain}). Return the number of nodes $v$ such that $\occ_S[v] > \occ_T[v]$.
\end{minipage}\\[3pt]
\hline
\begin{minipage}[t]{5.6cm}
For strings $S$ and $T$ find the number of equal palindromes, i.e., of triples $(i, j, k)$ such that $S[i..i{+}k] = T[j..j{+}k]$ is a palindrome.
\end{minipage}
&
\begin{minipage}[t]{6.2cm}
Build $\eertree(S,T)$, computing the values $\occ_S$ and $\occ_T$ in its nodes. The answer is $\sum_v \occ_S[v] \cdot \occ_T[v]$.
\end{minipage}\\
\hline
\end{tabular}
}

\subsection{Coping with deletions} \label{ssec:del}

In the proof of Proposition~\ref{nlogsigma}, an $O(n \log \sigma)$ algorithm for building an eertree is given. Nevertheless, in some cases one call of $\add$ requires $\Omega(n)$ time, and this kills some possible applications. For example, we may want to support an eertree for a string which can be changed in two ways: by appending a symbol on the right ($\add(c)$) and by deleting the last symbol ($\pop()$). Consider the following sequence of calls:
$$
\underbrace{\add(a), \ldots, \add(a),}_{n/3 \text{ times}}
\underbrace{\add(b), \pop(), \add(b), \pop(), \ldots, \add(b), \pop()}_{n/3 \text{ times}}
$$
Since each appending of $b$ requires $n/3$ suffix link transitions, the algorithm from Proposition~\ref{nlogsigma} will process this sequence in $\Omega(n^2)$ time independent of the implementation of the operation $\pop()$.

Below we describe two algorithms which build eertrees in a way that provides an efficient solution to the problem with deletions. 

\paragraph*{Searching suffix-palindromes with quick links.}
Consider a pair of nodes $v, \link[v]$ in an eertree and the symbol $b=v[|v|-|\link[v]|]$ preceding the suffix $\link[v]$ in $v$. In addition to the suffix link, we define the \emph{quick link}: let $\quick[v]$ be the longest suffix-palindrome of $v$ preceded in $v$ by a symbol different from $b$. 

\begin{lemma}
As a node $v$ is created, the link $\quick[v]$ can be computed in $O(1)$ time.
\end{lemma}

\begin{proof}
The two longest suffix-palindromes of $v$ are $u=\link[v]$ and $u'=\link[\link[v]]$. Assume that $v$ has suffixes $bu$ and $cu'$. If $c\ne b$, then $\quick[v]=u'$ by definition. If $c=b$, then clearly $\quick[v]=\quick[u]$. Thus we need a constant number of operations. The code computing the quick link of $v$ is given below. \qed
\end{proof}

\lstset{frame=tb,
  language=Pascal,
  aboveskip=3mm,
  belowskip=3mm,
  showstringspaces=false,
  columns=flexible,
  basicstyle=\ttfamily,
  numbers=none,
  keywordstyle=\bf,
  breaklines=true,
  breakatwhitespace=true
  tabsize=3
}
\begin{lstlisting}
if ( S[n - len[link[v]]] == S[n - len[link[link[v]]]] )
	quickLink[v] = quickLink[link[v]]
else 
	quickLink[v] = link[link[v]]
\end{lstlisting}

Recall that appending a letter $c$ to a current string $S$, we scan suffix-palindromes of $S$ to find the longest suffix-palindrome $Q$ preceded by $c$; then $\maxsuf(Sc)=cQc$. (If $cQc$ is a new palindrome, then this scan continues until $\link[cQc]$ is found.) The use of quick links reduces the number of scanned suffix-palindromes as follows. When the current palindrome is $v$, we check both $v$ and $\link[v]$. If both are not preceded by $c$, then all suffix-palindromes of $S$ longer than $\quick[v]$ are not preceded by $c$ either; so we skip them and check $\quick[v]$ next.

\begin{example} \label{ex:quick}
Let us call $\add(b)$ to the eertree of the string $S=aabaabaaba$. The longest suffix-palindrome of $S$ is the string $v=abaabaaba$. Since the symbols preceding $v$ and $\link[v]=abaaba$ in $S$ are distinct from $b$, we jump to $\quick[v]=a$, skipping the suffix-palindrome $aba$ preceded by the same letter as $\link[v]$. Now $\quick[v]$ is preceded by $b$, so we find $\maxsuf(Sb)=bab$. Note that $v$ ``does not know'' which symbol precedes its particular occurrence, and different occurrences can be preceded by different symbols. So there is no way to avoid checking the symbol preceding $\link[v]$.
\end{example}

Constructing an eertree with quick links, on each step we add $O(1)$ time and space for maintaining these links and possibly reduce the number of processed suffix-palindromes. So the overall time and space bounds from Proposition~\ref{nlogsigma} are in effect. Let us estimate the number of operations per step. The statements on ``series'' of palindromes, analogous to the next proposition, were proved in several papers (see, e.g., \cite[Lemmas 5,6]{KRS15} and \cite[Lemma 5]{FGKK14}).

\begin{proposition}\label{quicklog}
In an eertree, a path consisting of quick links has length $O(\log n)$.
\end{proposition}

\begin{corollary}
The algorithm constructing an eertree using quick links spends $O(\log n)$ time and $O(1)$ space for any call to $\add$.
\end{corollary}

\paragraph*{Using direct links.}
Now we describe the fastest algorithm for constructing an eertree which, however, uses more than $O(1)$ space for creating a node. Still, the space requirements are quite modest, so the algorithm is highly competitive:

\begin{proposition}
There is an algorithm which constructs an eertree spending $O(\log \sigma)$ time and $O(\min(\log\sigma, \log \log n))$ space for any call to $\add$. 
\end{proposition}

\begin{proof}
For each node we create $\sigma$ \emph{direct links}: $\direct[v][c]$ is the longest suffix-palindrome of $v$ preceded in $v$ by $c$. 

Let $Q$ be the longest suffix-palindrome of a string $S$, preceded by $c$ in $S$. Then either $Q=\maxsuf(S)$ or $Q=\direct[\maxsuf(S)][c]$, and the longest suffix-palindrome of $Q$, preceded by $c$, is $\direct[Q][c]$. Thus, we scan suffixes in constant time, and the time per step is now dominated by $O(\log\sigma)$ for searching an edge in the dictionary plus the time for creating direct links for a new node. 

Note that the arrays $\direct[v]$ and $\direct[\link[v]]$ coincide for all symbols except for the symbol $c$ preceding $\link[v]$ in $v$. Hence, creating a node $v$ we first find $\link[v]$, then copy $\direct[\link[v]]$ to $\direct[v]$ and assign $\direct[v][c]=\link[v]$. However, storing or copying direct links explicitly would cost a lot of space and time. So we do this implicitly, using fully persistent balanced binary search tree (\emph{persistent tree} for short; see\cite{DSST89}).  We will not fall into details of the internal of the persistent tree, taking it as a blackbox. The persistent tree provides full access to any of $m$ its \emph{versions}, which are balanced binary search trees. The versions are ordered by the time of their creation. An update of any version results in creating a new $(m{+}1)$th version, which is also fully accessible; the updated version remains unchanged. Such an update as adding a node or changing the information in a node takes $O(\log k)$ time and space, where $k$ is the size of the updated version. 

We store direct links from all nodes of the eertree in a single persistent tree. Each version corresponds to a node. Direct links $\direct[v][c]$ in a version $v$ are stored as a search tree, with the letter $c$ serving as the key for sorting (we assume an ordered alphabet). Creation of a node $v$ requires an update of the version corresponding to the node $\link[v]$. It remains to estimate the size of a single search tree. It is at most $\sigma$ by definition, and it is $O(\log n)$ by Proposition~\ref{quicklog}. Thus, the update time and space is $O(\min(\log \sigma, \log \log n))$, as required.\qed
\end{proof}

\paragraph*{Comparing different implementations.}\label{versions}

The three methods of building an eertree are gathered in the following table.\\[3pt]
\centerline{
\begin{tabular}{c|c|c|c}
Method
&
Time for $n$ calls
&
Time for one call
&
Space for one node\\
\hline
basic & $\Theta(n \log \sigma)$ & $\Omega(\log\sigma)$ but $O(n)$ & $\Theta(1)$\\
quickLink & $\Theta(n \log \sigma)$ & $\Omega(\log\sigma)$ but $O(\log n)$ & $\Theta(1)$ \\
directLink & $\Theta(n \log \sigma)$ & $\Theta(\log\sigma)$ & $O(\min(\log\sigma, \log \log n))$\\
\hline
\end{tabular} }

\smallskip
The basic version is the simplest one and uses the smallest amount of memory. Quick and direct links work somewhat faster, but their main advantage is that any single call is cheap, and thus can be reversed without much pain. Hence, one can easily maintain an eertree for a string with both operations $\add(c)$ and $\pop()$. Indeed, let $\add(c)$ push to a stack the node containing $P=\maxsuf(Sc)$ and, if $P$ is a new palindrome, the node containing $Q$ such that $P=cQc$. This takes $O(1)$ additional time and space. Then $\pop()$ reads this information from the stack and restores the previous state of the eertree in constant time.

The table above also suggests the question whether some further optimization of the obtained algorithms is possible.

\begin{question}
Is there an online algorithm which builds an eertree spending $O(\log \sigma)$ time and $O(1)$ space for any call to $\add$? 
\end{question}

\subsection{Enumerating rich strings}

By Lemma~\ref{nnodes}, the number of distinct subpalindromes in a length $n$ string is at most $n$. Such strings with exactly $n$ palindromes are called \emph{rich}. Rich strings possess a number of interesting properties; see, e.g., \cite{DJP01,GJWZ09}. The sequence A216264 in the Online Encyclopedia of Integer Sequences \cite{Slo} is the growth function of the language of binary rich strings, i.e., the $n$th term of this sequence is the number of binary rich strings of length $n$. J. Shallit computed this function up to $n=25$, thus enumerating several millions of rich strings. Using the results of Sect.~\ref{ssec:del}, we were able to raise the upper bound to $n=60$, enumerating several \emph{trillions} of rich strings in 10 hours on an average laptop. The new numerical data shows that this sequence grows much slower than it was expected before.

Proposition~\ref{pro:rich} below serves as the theoretic basis for such a breakthrough in computation. It is based on the following obvious corollary of Lemma~\ref{nnodes}.

\begin{lemma}\label{rich}
Any prefix of a rich string is rich.
\end{lemma}

\begin{proposition} \label{pro:rich}
Suppose that $R$ is the number of $k$-ary rich strings of length $\le n$, for some fixed $k$ and $n$. Then the trie built from all these strings can be traversed in time $O(R)$. 
\end{proposition}

\begin{proof}
For simplicity, we give the proof for the binary alphabet. The extension to an arbitrary fixed alphabet is straightforward. Consider the following code, using an eertree on a string with deletions.

\begin{lstlisting}
void calcRichString(i) 
	ans[i]++
	if (i < n)
		if (add('0') )
			calcRichString(i + 1)
		pop()
		if (add('1') )
			calcRichString(i + 1)
		pop()
\end{lstlisting}

Here $i$ is the length of the currently processed rich string. Recall that $\add(c)$ appends $c$ to the current eertree and returns the number of new palindromes, which is 0 or 1. Hence the modified string is rich if and only if $\add$ returns 1. Note that any added symbol will be deleted back with $\pop()$. So we exit every particular call to ${\sf calcRichString}$ with the same string as the one we entered this call. As a result, the call ${\sf calcRichString}(0)$ traverses depth-first the trie of all binary rich strings of length $\le n$. 

As was mentioned in Sect.~\ref{ssec:del}, the $\pop$ operation works in constant time. For $\add$ we use the method with direct links. Since the alphabet is constant-size, the array $\direct[v]$ can be copied in $O(1)$ time. Hence, $\add$ also works in $O(1)$ time. The number of $\pop$'s equals the number of $\add$'s, and the latter is twice the number of rich strings of length $<n$. The number of other operations is constant per call of ${\sf calcRichString}$, so we have the total $O(R)$ time bound.\qed 
\end{proof}

\begin{remark}
Visit \texttt{http://pastebin.com/4YJxVzep} for an implementation of the above algorithm. In 10 hours, it computed the first 58 terms of the sequence A216264. To increase the number of terms to 60, we used a few optimization tricks which reduce the constant in the $O$-term. We do not discuss these tricks here, because they make the code less readable.
\end{remark}

\subsection{Persistent eertrees}

In Sect.~\ref{ssec:del} we build an eertree supporting deletions from a string. A natural generalization of this approach leads to \emph{persistent} eertrees. Recall that a persistent data structure is a set of ``usual'' data structures of the same type, called \emph{versions} and ordered by the time of their creation. A call to a persistent structure asks for the access or update of any specific version. Existing versions are neither modified nor deleted; any update creates a new (latest) version.
 
Consider a \emph{tree of versions} $\cT$ whose nodes, apart from the root, are labeled by symbols. The tree represents the set of versions of some string $S$: each node $v$ represents the string read from the root to $v$. Recall that we denote a node of a data structure by the same letter as the string related to it. Note that some versions can be identical except for the time of their creation (i.e., for the number of a node). 
The problem we study is maintaining an eertree for each version of $S$. More precisely, the function $\addv(v,c)$ to be implemented adds a new child $u$ labeled by $c$ to the node $v$ of $\cT$ and computes $\eertree(u)$. The data structure which performs the calls to $\addv$, supporting the eertrees for all nodes of $\cT$, will be called a \emph{persistent eertree}. Surprisingly enough, this complicated structure can be implemented efficiently in spite of the fact that the current string cannot be addressed directly for symbol comparisons.

\begin{proposition}
The persistent eertree can be implemented to perform each call to $\addv(v,c)$ in $O(\log|v|)$ time and space.
\end{proposition}

\begin{proof}
We use the method with direct links and build, as in Sect.~\ref{common}, a joint eertree for all versions. Each node of the tree $\cT$ stores links to the palindromes of the corresponding version of $S$. Overall, the node $v$ of $\cT$ contains the following information: a binary search tree $\search[v]$, containing links to all subpalindromes of $v$; link $\maxsuf[v]$ to the maximal suffix-palindrome of $v$; array $\pred[v]$, whose $i$th element is the link to the predecessor $z$ of $v$ such that the distance between $z$ and $v$ in $\cT$ is $2^i$ ($i\ge 0$); and the symbol $\symb[v]$ added to the parent of $v$ to get $v$. All listed parameters except for $\search[v]$ use $O(\log|v|)$ space. For search trees we use, as in Sect.~\ref{ssec:del}, the persistent tree \cite{DSST89}, reducing both time and space for copying the tree and inserting one element to  $O(\log|v|)$. (Recall that another persistent tree is used inside the eertree for storing direct links of all nodes.)

Now we implement $\addv(v,c)$ in time $O(\log|v|)$. Note that for any $i$ the symbol $v[i]$ can be found in $O(\log|v|)$ time. Indeed, this symbol is $\symb[z]$, where $z$ is the predecessor of $v$ such that the distance between $z$ and $v$ is $h = |v|-i$. Using the binary representation of $h$, we can reach $z$ from $v$ in at most $\log |v|$ steps following the appropriate $\pred$ links. 

Let $V$ be the current number of versions (at any time). Creating a new version $u$ with the parent $v$, we increment $V$ by one and compute all parameters for $u$. First we compute $\pred[u]$. This can be done in $O(\log|v|)$ time because $\pred[u][0]=v$ and $\pred[u][i]=\pred[\pred[u][i-1]][i-1]$ for $i>0$.

To compute the palindrome $y=\maxsuf[u]$, we call $\add(c)$ for the string $v$. Let $x$ be the parent of $y$ in the eertree. Then $x=\maxsuf[v]$ if $\maxsuf[v]$ is preceded by $c$ in $v$ and $x=\direct[\maxsuf[v]][c]$ otherwise. Hence, to compute $y$ we access exactly one symbol of $v$. Further, if $y$ is not in the eertree, a new node of the eertree should be created for $y$. It is easy to see that $\link[y]=\go[\direct[x][c]][c]$. Next, $\direct[u]$ is copied from $\direct[\link[u]]$, with one element replaced by $\link[u]$. To find this element, we need to know the letter of $v$ preceding $x$. Therefore, to find $\maxsuf[u]$ and modify eertree if necessary, we need $O(\log|v|)$ time for accessing a constant number of symbols in $v$ and $O(\log\sigma)$ time for the rest of computation in $\add(c)$. Finally, we create a version of the search tree for $u$, updating the version for $v$ with $y$ (if $y$ is in the search tree for $v$, this tree is copied to the new version without changes). This operation takes $O(\log|v|)$ as well. The proposition is proved. The code for $\addv(v,c)$ is given below.\qed
\end{proof}

\begin{lstlisting}[escapeinside={/*@}{@*/}]
void getpred(v, par)
	pred[v][0] = par
	i = 1
	while (pred[v][i] > 0)
		pred[v][i + 1] = pred[ pred[v][i] ][i]
		i++
int addVersion(v, c)
	t++ // /*@the number of versions, initialized by 0@*/
	u = t	
	symb[u] = c
	pred[u] = getpred(u, v)
	if (c == v[len[v] - len[maxSuf[v]]])
		x = maxSuf[v] 
	else 
		x = directLink[maxSuf[v]][c]
	maxSuf[u] = /*@to[x][c] //created if does not exist@*/
	searchTree[u] = insert(searchTree[v], maxSuf[u])
	return u
\end{lstlisting}

\section{Factorizations into Palindromes}\label{splittingSection}

As was mentioned in the introduction, the $k$-factorization problem can be solved online in $O(kn)$ time for the length $n$ string and any $k$ \cite{KRS15}. 
In this section we are aimed at solving this problem in time independent of $k$. This setting is motivated by the fact that the expected palindromic length of a random string is $\Omega(n)$ \cite{Rav03}, and the $O(kn)$ asymptotics is quite bad for such big values of $k$. On the positive side, the palindromic length of a string $S$, which is the minimum $k$ such that a $k$-factorization of $S$ exists, can be found in $O(n\log n)$ time \cite{FGKK14}.

\subsection{Palindromic length vs $k$-factorization}

\begin{lemma} \label{kplus2}
Given a $k$-factorization of a length $n$ string $S$, it is possible, in $O(n)$ time, to factor $S$ into $k{+}2t$ palindromes for any positive integer $t$ such that $k{+}2t\le n$.
\end{lemma}

\begin{proof}
Let $P_1,\ldots,P_k$ be palindromes, $S=P_1\cdots P_k$, $k\le n-2$. It is sufficient to show how to factor $S$ into $k{+}2$ palindromes. If $|P_i|\ge 3$ for some $i$, then we split $P_i$ into three palindromes: the first letter, the last letter, and the remaining part. Otherwise, there are some $P_i$, $P_j$ of length 2, each of which can be split into two palindromes.\qed
\end{proof}

Thus, $k$-factorization problem is reduced in linear time to two similar problems: factor a string into the minimum possible odd (resp. even) number of palindromes. We solve these two problems using an eertree. To do this, we first describe an algorithm, based on an eertree and finding the palindromic length in time $O(n \log n)$. While its asymptotics is the same as of the algorithm of \cite{FGKK14}, its constant under the $O$-term is much smaller (see Remark~\ref{ficiConst} below) and its code is simpler and shorter.

\begin{proposition}
Using an eertree, the palindromic length of a length $n$ string can be found online in time $O(n\log n)$. 
\end{proposition}

\begin{proof}
For a length $n$ string $S$ we compute online the array $\ans$ such that $\ans[i]$ is the palindromic length of $S[1..i]$. Note that any $k$-factorization of $S[1..i]$ can be obtained by appending a suffix-palindrome $S[j{+}1..i]$ of $S[1..i]$ to a $(k{-}1)$-factorization of $S[1..j]$. Thus, $\ans[i]=1+\min\{\ans[j]\mid S[j+1..i] \text{ is a palindrome}\}$.

To compute $\ans$ efficiently, we store two additional parameters in the nodes of the eertree: \emph{difference} $\diff[v] = \len[v] - \len[\link[v]]$ and \emph{series link} $\series[v]$, which is the longest suffix-palindrome of $v$ having the difference unequal to $\diff[v]$. Series links are similar to quick links, which are not suitable for the problem studied. Clearly, the difference is computable in $O(1)$ time and space on the creation of a node; the following code shows that the same is true for the series link.

\begin{lstlisting}
if (diff[v] == diff[link[v]])
	seriesLink[v] = seriesLink[link[v]]
else 
	seriesLink[v] = link[v]
\end{lstlisting}

The following ``naive'' implementation computes $\ans[n]$ in $O(n)$ time.
\begin{lstlisting} [escapeinside={/*@}{@*/}]
ans[n] = /*@$\infty$@*/ 
for (v = maxSuf; len[v] > 0; v = link[v])
	ans[n] = min(ans[n], ans[n - len[v]] + 1)
\end{lstlisting}

With series links, the same idea can be rewritten as follows:
\begin{lstlisting} [escapeinside={/*@}{@*/}]
int getMin(u)
	res = /*@$\infty$@*/ 
	for (v = u; len[v] > len[seriesLink[u]]; v = link[v])
		res = min(res, ans[n -	 len[v]] + 1)
	return res
ans[n] = /*@$\infty$@*/ 
for (v = maxSuf; len[v] > 0; v = seriesLink[v])
	ans[n] = min(ans[n], getMin(v))
\end{lstlisting}

The $\getmin$ function has linear-time worst-case complexity, and we are going to speed it up to a constant time. By the \emph{series} of a palindrome $u$ we mean the sequence of nodes in the suffix path of $u$ from $u$ (inclusive) to $\series[u]$ (exclusive). Note that $\getmin[u]$ loops through the series of $u$. Comparing $\diff[u]$ and $\diff[\link[u]]$, we can check whether the series of $u$ contains just one palindrome. If this is the case, then $\res=\ans[n-\len[u]]+1$ can be computed in $O(1)$ time. Hence, below we are interested in series of at least two elements. A suffix-palindrome $u$ of $S$ is called \emph{leading} if either $u=\maxsuf(S)$ or $u=\series[v]$ for some suffix-palindrome $v$ of $S$. We need four auxiliary lemmas.

\begin{lemma}\label{intersect}
If a palindrome $v$ of length $l\ge n/2$ is both a prefix and a suffix of a string $S[1..n]$, then $S$ is a palindrome.
\end{lemma}

\begin{proof}
Let $i \le n/2$. Then $ S[i]=v[i]=v[l-i+1]=S[n-i+1] $, i.e., $S$ is a palindrome by definition.\qed
\end{proof}

\begin{lemma}\label{prevocc}
Suppose $v$ is a leading suffix-palindrome of a string $S[1..n]$ and $u=\link[v]$ belongs to the series of $v$. Then $u$ occurs in $v$ exactly two times: as a suffix and as a prefix.
\end{lemma}

\begin{proof}
Let $i = n - |v| + 1$. Then $v = S[i..n], u=S[i{+}\diff[v]..n]=S[i..n{-}\diff[v]]$. Since $\diff[u]=\diff[v]$, we have $\diff[v]\le |v|/2$, so that the two mentioned occurrences of $u$ touch or overlap. If there exist $k,t$ such that $i<k<i{+}\diff[v]$ and $S[k..t] = u$, then $S[k..n]$ is a palindrome by Lemma~\ref{intersect}. This palindrome is a proper suffix of $v$ and is longer than $\link[v]$, which is impossible. \qed
\end{proof}

\begin{lemma} \label{maxseries} 
Suppose $v$ is a leading suffix-palindrome of a string $S[1..n]$ and $u=\link[v]$ belongs to the series of $v$. Then $u$ is a leading suffix-palindrome of $S[1..n{-}\diff[v]]$.
\end{lemma}

\begin{proof}
If $u$ is not leading, then the string $S[1..n{-}\diff[v]]$ has a suffix-palindrome $z=S[j..n{-}\diff[v]]$ with $\link[z]=u$ and $\diff[z]=\diff[u]$. Since $u$ is both a prefix and a suffix of $z$ and $|z|=|v|\le 2|u|$, clearly $z=v$. Then $w=S[j..n]$ is a palindrome by Lemma~\ref{intersect}. Assume that $w$ has a suffix-palindrome $v'$ which is longer than $v$. Then $v'$ begins with $u$, and this occurrence of $u$ is neither prefix nor suffix of $z=S[j..n{-}\diff[v]]$, contradicting Lemma~\ref{prevocc}. Therefore, $v=\link[w]$ and $\diff[w]=\diff[v]$, which is impossible because $v$ is leading. This contradiction proves that $u$ is leading.\qed
\end{proof}

\begin{lemma} \label{seriesway}
In an eertree, a path consisting of series links has length $O(\log n)$.
\end{lemma}

\begin{proof}
Follows from \cite[Lemma 6]{KRS15}, since any leading suffix-palindrome is also leading in terms of \cite{KRS15}. 
\end{proof}

By Lemma~\ref{seriesway}, the function $\ans(n)$ calls $\getmin$ $O(\log n)$ times. Now consider an $O(1)$ time implementation of $\getmin$. Recall that it is enough to analyze non-trivial series of palindromes; they look like in Fig.~\ref{fig:series}. The first positions of all palindromes in the depicted series of $v$ and $\link[v]$ match (because $\diff[v]=\diff[\link[v]]$) except for the last palindrome in the series of $v$.

\begin{figure}[hbt!]
\vspace*{-4mm}
\centering
\includegraphics[trim=140 625 140 129,clip]{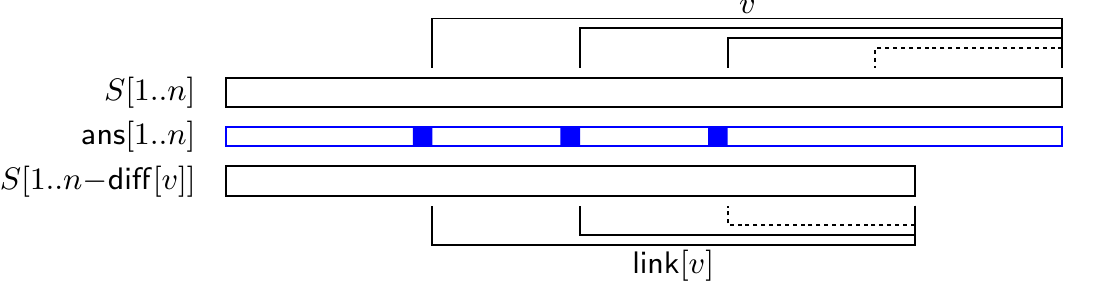}
\vspace*{-2mm}
\caption{Series of a palindrome $v$ in $S[1..n]$ and of $\link[v]$ in $S[1..n{-}\diff[v]]$. Leading palindromes of the next series are shown by dash lines. The function $\getmin(v)$ returns the minimum of the values of $\ans$ in the marked positions, plus one.}
\label{fig:series}
\end{figure}
\vspace*{-3mm}

We see that $\diff[v]$ steps before we already computed the minimum of all but one required numbers. If we memorize the minimum at that moment, we can use it now to obtain $\getmin$ in constant time. We store such a minimum as an additional parameter $\ddp$ of the node of the eertree, updating it each time the palindrome represented by a node becomes a leading suffix-palindrome. Lemmas~\ref{prevocc} and~\ref{maxseries} ensure that when we access $\ddp[\link[v]]$ to compute $\getmin[v]$, it is exactly the value computed $\diff[v]$ steps before. The computations with $\ddp$ can be performed inside the $\getmin$ function:

\begin{lstlisting}
int getMin(v)
	dp[v] = ans[n - (len[seriesLink[v]] + diff[v])] //last 
	if (diff[v] == diff[link[v]])// non-trivial series
		dp[v] = min(dp[v], dp[link[v]])
	return dp[v] + 1
\end{lstlisting}

Here $\ddp[v]$ is initialized by the value of $\ans$ in the position preceding the last element of the series of $v$. It is nothing to do if this series does not have other elements; if it has, the minimum value of $\ans$ in the corresponding positions is available in $\ddp[\link[v]]$.
\end{proof}

\begin{remark} \label{seriestree}
Series links can replace quick links in the construction of eertrees. Recall that in the method of quick links (Sect.~\ref{ssec:del}) after checking the symbols in $S$ preceding $v$ and $\link[v]$ we assign $\quick[v]$ to $v$ and repeat the process until the required symbol is found or the node $-1$ is reached. With series links, the termination condition is the same, but the process is a bit different. We first put $v=\maxsuf(S)$ and check the symbol before $v$. Then we keep repeating two operations: check the symbol preceding $\link[v]$ and assign $\series[v]$ to $v$. In this way, all ``skipped'' symbols, including the symbol preceding $v$, equal the symbol preceding the previous value of $\link[v]$. (This is due to periodicity of $v$; for details see, e.g., \cite[Sect.\,2]{KRS15}.) The number of iterations of the cycle equals the number of series of suffix-palindromes of $S$, which is $O(\log n)$ by Lemma~\ref{seriesway}.
\end{remark}

\begin{remark}\label{ficiConst}
Let $t_i$ be the number of series of suffix-palindromes for the string $S[1..i]$. Our computation of palindromic length\footnote{See {\tt http://ideone.com/xE2k6Y} for an implementation.} performs, on each step, the following operations. For the eertree: at most $t_i{+}1$ symbol comparisons (Remark~\ref{seriestree}) and one ($\log\sigma$)-time access to a dictionary. For palindromic length: $t_i$ calls to $\getmin$, which fills one cell in $\ddp$ and one cell in $\ans$.

The algorithm by Fici et al. \cite[Figure 8]{FGKK14} on each step builds three arrays ($G, G', G''$), each containing $t_i$ triples of numbers; totally $9t_i$ cells to be filled. So, our algorithm should work significantly faster. 
\end{remark}

Now we return to the $k$-factorization problem.

\begin{proposition}
Using an eertree, the $k$-factorization problem for a length $n$ string can be solved online in time $O(n\log n)$. 
\end{proposition}

\begin{proof}
The above algorithm for palindromic length can be easily modified to obtain both minimum odd number of palindromes and minimum even number of palindromes needed to factor a string. Instead of $\ans$ and $\ddp$, one can maintain in the same way four parameters: $\ans_o$, $\ans_e$, $\ddp_o$, $\ddp_e$, to take parity into account. Now $\ans_o$ (resp., $\ans_e$) uses $\ddp_e$ (resp., $\ddp_o$), while $\ddp_o$ (resp., $\ddp_e$) uses $\ans_o$ (resp., $\ans_e$). The reference to Lemma~\ref{kplus2} finishes the proof.
\end{proof}

\subsection{Towards a linear-time solution}

A big question is whether palindromic length can be found faster than in $O(n\log n)$ time. First of all, it may seem that the bound $O(n\log n)$ for our algorithm is imprecise. Indeed, for building an eertree we scan only $O(n)$ suffix palindromes even when we use just suffix links (see the proof of Proposition~\ref{nlogsigma}). For palindromic length, on each step we run through all suffix-palindromes, but possibly skipping many of them due to the use of series links. Can this number of scanned palindromes be $O(n)$ as well? As was observed in \cite{FGKK14}, the answer is ``yes'' on average, but ``no'' in the worst case: processing any length $n$ prefix of the famous \emph{Zimin word}, one should analyze  $\Theta(n\log n)$ series of palindromes (all of them 1-element, but this does not help).

Below we design an $O(n)$ offline algorithm for building an eertree of a length $n$ string $S$ over the alphabet $\{1,\ldots,n\}$, getting rid of the $\log\sigma$ factor in online algorithms. Then we discuss ideas which may help to obtain the palindromic length from an eertree in linear time. The offline algorithm consists of four steps.


\noindent
1. Using Manacher's algorithm, compute arrays $\sf oddR$ and $\sf evenR$, where  ${\sf oddR}[i]$ (resp. ${\sf evenR}[i]$) is the radius of the longest subpalindrome of $S$ with the center $i$ (resp., $i{+}1/2$).\\[2pt]
2. Compute the longest and the second longest suffix-palindromes for any prefix of $S$. We use variables $\ell,\ell'$, and $r$ such that after $r$th iteration the string $S[\ell..r]$ (resp., $S[\ell'..r]$) is the longest (resp., second longest) suffix-palindrome of $S[1..r]$.  

\begin{lstlisting}[escapeinside={/*@}{@*/}]
/*@$\ell$@*/ = 2
for (r = 1; r /*@$\le$@*/ n; r++)
	/*@$\ell$@*/--
	while ( !isPal(S[/*@$\ell$@*/..r] )
		/*@$\ell$@*/++
	/*@$\ell'$@*/ = max(/*@$\ell'$@*/ - 1 , /*@$\ell$@*/ + 1)
	while ( !isPal(S[/*@$\ell'$@*/..r] ) &&  (/*@$\ell'\ {\le}$@*/ r) )
		/*@$\ell'$@*/++
	C[(/*@$\ell$@*/ + r) / 2].push(1, r)
	C[(/*@$\ell'$@*/ + r) / 2].push(2, r)
\end{lstlisting}

The function $\sf isPal$, checking whether a given substring is a palindrome, works in $O(1)$ time, using the value obtained on step 1 for the center $(\ell{+}r)/2$. Each element of the array $C$ is a connected list; the indices are both integers and half-integers. The internal cycles make at most $2n$ increments of each of the variables $\ell$, $\ell'$; hence, the whole step works in linear time.\\[2pt]
3. Build the suffix array $SA$ and the $LCP$ array for $S$; for the alphabet $\{1,\ldots,n\}$, this can be done in linear time \cite{PST07}. Recall that $LCP[i]$ is the length of the longest common prefix of $S[SA[i]..n]$ and $S[SA[i{-}1]..n]$.\\[2pt] 
4. Recall from Sect.~\ref{ssec:prop} that an eertree consists of two tries, containing right halves of odd-length and even-length palindromes, respectively. Build each of them using a variation of the algorithm, constructing a suffix tree from a suffix array and its $LCP$ array \cite{KLAAP01}. The algorithm for odd-length palindromes is given below; the algorithm for even lengths is essentially the same, so we omit it.

\begin{lstlisting}[escapeinside={/*@}{@*/}]
path = (-1) // /*@stack for the current branch of the trie@*/
for (i = 1; i /*@$\le$@*/ n; i++)
	k = SA[i] // start processing palindromes centered at k
	while (path.size() > LCP[i] + 1)
		path.pop()
	for (j = path.size(); j /*@$\le$@*/ oddR[k]; j++) //can be empty
		path.push( newNode(path.top(), S[k + j - 1]) )
	for (j = 1; j /*@$\le$@*/ C[k].size(); j++)
		(rank, r) = C[k][j]
		node[rank][r] = path[r - k + 1]
\end{lstlisting}

The function $\newv(v,a)$ returns a new node attached to the node $v$ with the edge labeled by $a$. Array $\nod[1][1..n]$ (resp., $\nod[2][1..n]$) contains links to the longest (resp., second longest) palindromes ending in given positions. Now estimate the working time of this algorithm. The outer cycle works $O(n)$ time plus the time for the inner cycles. The number of pop operations is bounded by  the number of pushes, and the latter is the same as the number of nodes in the resulting eertree, which is $O(n)$. The total number of iterations of the third inner cycle is the number of palindromes stored in the whole array $C$; this is exactly $2n$, see step 2. Thus, the algorithm works in $O(n)$ time. 

After running both the above code and its modification for even-length palindromes, we obtain the eertree without suffix links and the arrays $\nod[1]$, $\nod[2]$. From these arrays the suffix links can be computed trivially:

\begin{lstlisting}[escapeinside={/*@}{@*/}]
for (i = 1; i /*@$\le$@*/ n; i++)
	link[ node[1][i] ] = node[2][i];
\end{lstlisting}

Thus we have proved

\begin{proposition}
The eertree of a length $n$ string over the alphabet $\{1,\ldots,n\}$ can be built offline in $O(n)$ time. 
\end{proposition}

Now return to the palindromic length. Even with an $O(n)$ preprocessing for building the eertree, we still need $O(n\log n)$ time for factorization. Note that in \cite{KRS15} an $O(k n \log n)$ algorithm for $k$-factorization was transformed into a $O(kn)$ algorithm using bit compression (the so-called \emph{method of four Russians}). That algorithm produced a $k\times n$ bit matrix (showing whether a $j$th prefix of the string is $i$-factorable), so such a speed up method was natural. In our case we work with integers, so the direct application of a bit compression is impossible. However, we have the following property.

\begin{lemma}\label{splitting}
If $S$ is a string of palindromic length $k$ and $c$ is a symbol, then the palindromic length of $Sc$ is $k{-}1, k$, or $k{+}1$.
\end{lemma}

\begin{proof}
Any $k$-factorization of $S$ plus the substring $c$ give a $(k{+}1)$-factorization of $Sc$. Suppose $Sc$ has a $t$-factorization $P_1\cdots P_t$ for a smaller $t$. Then $P_t=Pc$ has length $>1$. Hence, either $P=c$ and $S$ has the $t$-factorization $P_1\cdots P_{t-1}c$ or $P=cQ$ for a palindrome $Q$ and $S$ has the $(t{+}1)$-factorization $P_1\cdots P_{t-1}cQ$. The result now follows. \qed
\end{proof}

Consider a $n\times n$ bit matrix $M$ such that $M[i,j]=1$ if and only if $S[1..j]$ is $i$-factorable. For $j$th column, we have to compute just two values: in the rows $k{-}1$ and $k$, where $k$ is the palindromic length of $S[1..j{-}1]$ (if $M[k{-}1,j]=M[k,j]=0$, we write $M[k{+}1,j]=1$ by Lemma~\ref{splitting}). For each value we should apply the OR operation to $\log n$ bit values, to the total of $2n \log n$ bit operations. If we will be able to arrange these operations naturally in groups of size $\log n$, we will use the bit compression to get just $O(n)$ operations. So we end the paper with the following conjecture.

\begin{conjecture}
Using Lemma~\ref{splitting}, eertree and the method of four Russians it is possible to find palindromic length of a string in $O(n\log\sigma)$ time online and in $O(n)$ time offline.
\end{conjecture}

\section{Conclusion}

In this paper, we proposed a new tree-like data structure, named eertree, which stores all palindromes occurring inside a given string. The eertree has linear size (even sublinear on average) and can be built online in nearly linear time. We proposed some advanced modifications of the eertree, including the joint eertree for several strings, the version supporting deletions from a string, and the persistent eertree.

Then we provided a number of applications of the eertree. The most important of them are the new online algorithms for $k$-factorization, palindromic length, the number of distinct palindromes, and also for computing the number of rich strings up to a given length.

For further research we formulated a conjecture on the linear-time factorization into palindromes and an open problem about the optimal construction of the eertree.

\paragraph*{Acknowledgments.} The authors thank A.~Kul'kov, O.~Merkuriev and G.~Nazarov for helpful discussions. 

\bibliographystyle{splncs03}
\bibliography{my_bib}

\end{document}